\documentclass[journal]{IEEEtran}

\usepackage[pdftex]{graphicx}
\usepackage{amsmath}
\usepackage{amssymb}
\usepackage{upref}
\usepackage{amsfonts}
\usepackage{graphicx}
\usepackage{epic}
\usepackage{psfrag}
\usepackage{verbatim}

\usepackage{amsfonts}
\usepackage{amssymb}

\usepackage{color}
\usepackage{textcomp}

\newcommand{\deff}{\mbox{$\stackrel{\rm def}{=}$}}

\newcommand{\cE}{{\cal E}}

\newcommand{\cG}{{\cal G}}

\newcommand{\cS}{{\cal S}}

\newcommand{\cP}{{\cal P}}

\newcommand{\cT}{{\cal T}}
\newcommand{\cX}{{\cal X}}

\newcommand{\sP}{\cP}
\newcommand{\sG}{\cG}

\newcommand{\Gr}{\smash{{\sG\kern-1.5pt}_q\kern-0.5pt(n,k)}}
\newcommand{\Grr}{\smash{{\sG\kern-1.5pt}_q\kern-0.5pt(n,r)}}
\newcommand{\Gfourk}{\smash{{\sG\kern-1.5pt}_q\kern-0.5pt(4k,2k)}}
\newcommand{\Gk}{\smash{{\sG\kern-1.5pt}_q\kern-0.5pt(n,k_1)}}
\newcommand{\Gkk}{\smash{{\sG\kern-1.5pt}_q\kern-0.5pt(n,k_2)}}
\newcommand{\Grtwo}{\smash{{\sG\kern-1.5pt}_2\kern-0.5pt(n,k)}}
\newcommand{\Gkone}{\smash{{\sG\kern-1.5pt}_q\kern-0.5pt(n,k_1)}}
\newcommand{\Gktwo}{\smash{{\sG\kern-1.5pt}_q\kern-0.5pt(n,k_2)}}
\newcommand{\Ps}{\smash{{\sP\kern-2.0pt}_q\kern-0.5pt(n)}}

\newtheorem{theorem}{Theorem}
\newtheorem{cor}{Corollary}
\newtheorem{lemma}[theorem]{Lemma}

\newtheorem{remark}{Remark}
\newtheorem{example}{Example}

\newenvironment{algorithm}{%
       \begin{minipage}{\columnwidth}\vspace{0.5ex}%
       \makebox[0ex]{}\hrulefill\makebox[0ex]{}\\*\normalsize}{%
             \makebox[0ex]{}\hrulefill\makebox[0ex]{}\end{minipage}}

\begin{document}

\bibliographystyle{IEEEtran}

\title{Simple and Robust Binary Self-Location Patterns}
\author{Alfred~M.~Bruckstein,~Tuvi~Etzion,~\IEEEmembership{Fellow,~IEEE},~Raja~Giryes,~Noam~Gordon,~Robert~J.~Holt,~and~Doron~Shuldiner
\thanks{A. M. Bruckstein, T. Etzion, R. Giryes, N. Gordon, and D. Shuldiner
are with the Department of Computer Science, Technion
--- Israel Institute of Technology, Haifa 32000, Israel. (email:
freddy@cs.technion.ac.il, etzion@cs.technion.ac.il, raja@cs.technion.ac.il,  ngordon@cs.technion.ac.il,
doron.email@gmail.com).}
\thanks{R. J. Holt is with the Department of Mathematics and Computer Science,
Queensborough, City University of New York, Bayside, NY 11364, USA.
(email: RHolt@qcc.cuny.edu). }
\thanks{The work of T. Etzion was supported in part by the United States-Israel
Binational Science Foundation (BSF), Jerusalem, Israel, under
Grant No. 2006097.}
\thanks{The work of R. J. Holt was supported in part by
Bell Laboratories of Alcatel-Lucent.} }

\maketitle

\begin{abstract}
A simple method to generate a two-dimensional binary grid pattern,
which allows for absolute and accurate self-location in a finite
planar region, is proposed. The pattern encodes position
information in a local way so that reading a small number of its
black or white pixels at any place provides sufficient data from
which the location can be decoded both efficiently and robustly.
\end{abstract}


%*******************************************************************************
%*                                                                             *
%*                            1. Introduction                                  *
%*                                                                             *
%*******************************************************************************

\section{Introduction}

\PARstart{T}{ake} a blindfolded man on a random one-hour walk around town and
then remove his blindfold. How will he know where he is?  He has
several options, based on the information he can gather.  The man
could carefully count his steps and take note of every turn during
the blindfolded walk to know his location relative to the
beginning of his trip.  Armed with a navigation tool such as a
sextant or GPS unit, he could ask the stars or the GPS satellites
where he is. Lastly, he could simply look around for a reference,
such as a street sign, a landmark building, or even a city map
with a little arrow saying ``\emph{You are here}.''

There are numerous applications where a similar problem is
encountered. We need to somehow measure the position of a mobile
or movable device, using some sort of sensory input.  Wheeled
vehicles can count the turns of their wheels much like the man
counting his steps.  Similarly, many devices, from industrial
machine stages to ball-mice, employ sensors which are coupled with
the mechanics and count small physical steps of a known length, in
one or more dimensions. The small relative position differences
can be accumulated to achieve \emph{relative self-location} to a
known starting point.  More recent technologies, such as those
found in modern optical mice, use imaging sensors instead of
mechanical encoders to estimate the relative motion by constantly
inspecting the moving texture or pattern of the platform beneath
them.

Sometimes the inherent accumulating error in relative
self-location methods, or some other reasons, make them infeasible
or unfit for certain applications, where we would want the
capability to obtain instant and accurate \emph{absolute
self-location}.  Given several visible landmarks of known
locations, a mobile robot could calculate its position through a
triangulation \cite{cohen92}.  Alternatively, cleverly designed
space fiducials (e.g., \cite{Bruckstein2000}), whose appearance
changes with the angle of observation, can also serve for
self-location.

Much like street signs for people, there are absolute
self-location methods that provide sufficient local information to
the device sensors, such that the absolute positioning can be
attained. Specifically, planar patterns have been suggested, where
a small local sample from anywhere in the pattern provides
sufficient information for decoding the absolute position. A naive
example could consist of a floor filled with densely packed
miniature markings, in which the exact coordinates are literally
inscribed inside each marking. Of course, that would require a
high sensor resolution and character recognition capabilities.
Indeed, there are much more efficient methods, which do with
considerably less geometric detail in the pattern. Some commercial
products have been utilizing this approach, e.g., a pen with a
small imaging device in its tip, writing on paper with a special
pattern printed on it, which allows full tracking of the pen
position at any time \cite{anoto-web}.

A classic method for absolute self-location in one dimension is
the use of de Bruijn sequences~\cite{deB46,Fre82}.  A de Bruijn
sequence of order $n$ over a given alphabet of size $q$ is a
cyclic sequence of length $q^n$, which has the property that each
possible sequence of length $n$ of the given alphabet appears in
it as a consecutive subsequence exactly once. Thus, sampling~$n$
consecutive letters somewhere in the sequence is sufficient for
perfect positioning of the sampled subsequence within the
sequence. Several methods for the construction of de Bruijn
sequences have been proposed,
e.g.~\cite{Fre82,Lem70,Ara76,EtLe84}. There is also a
two-dimensional generalization, i.e., it is possible to construct
a two-dimensional cyclic arrays in which each rectangular
sub-array of a certain size $k \times n$ appears exactly once in
the array. These types of arrays are called \emph{perfect maps},
e.g.~\cite{Etz88,Pat94} and they can serve as the basis for
absolute self-location on the plane.

Of special interest and importance in communication are
maximal-length linear shift-register sequences known also as
M-sequences or pseudo-noise sequences~\cite{Golomb}. An
\emph{M-sequence} of order $n$ is a sequence of length $2^n-1$
generated by a linear feedback shift-register of length $n$. In a
cyclic sequence of this type, each nonzero $n$-tuple appears
exactly once as a window of length $n$ in one period of the
sequence exactly once. These sequences have many important and
desired properties~\cite{Golomb,McSl76}. A two-dimensional
generalization of such sequences was presented in~\cite{McSl76}
and are called pseudo-random arrays. We note also that M-sequences
can be used for robust one-dimensional location by using their
error-correction properties, as analyzed in~\cite{KuWe92}.

In this work we propose a simple product construction to generate
a two-dimensional binary patterns for absolute self-location. The
paper is organized as follows. In Section~\ref{sec:construct} we
present the product construction based on two sequences with some
one-dimensional window properties. A two-dimensional array with
optimal self-location based on sensing a cross shape is obtained
by this construction. In Section~\ref{sec:robust} we prove that
the same construction can be used for reasonably effective
error-correction of self-location with a rectangular shape. Our
conclusions and some interesting problems for future research are
presented in Section~\ref{sec:robust}.

\section{The Proposed 2D Self-Location Pattern}
\label{sec:construct}

Our approach for building two-dimensional arrays with
self-location properties is based on a product of two sequences,
one of which is a de Bruijn sequence and the other being a
sequence in which only half of the patterns appear.

\subsection{Half de Bruijn Sequences}
A \emph{half de Bruijn sequence} of order $n$ is a (cyclic)
sequence of length $2^{n-1}$ which has the property that for each
possible $n$-tuple $X$, either $X$ or $\bar{X}$ (the bitwise
complement of $X$), but not both, appear in the sequence exactly
once as a subsequence.

There are many different ways to construct half de Bruijn
sequences. One method, in which a half de Bruijn sequence of
length $2^{n-1}$ is generated from a de Bruijn sequence of order
$n-1$ by using the inverse of the well known mapping {\bf D},
called the {\bf D}-morphism, is described in~\cite{Lem70}. Another
one is based on M-sequences. The following results were proved
in~\cite{Ara76}.
\begin{theorem}
If $\cS$ is an M-sequence of order $n-1$ then for each pair of
$n$-tuples $X$ and $\bar{X}$ either $X$ or $\bar{X}$ appears in
$\cS$, with the exception of the pair which consists of the
all-zero and all-one $n$-tuples.
\end{theorem}
\begin{cor}
Let $\cS$ be an M-sequence of order $n-1$ and let $\cS'$ be the
sequence obtained from $\cS$ by adding another \emph{one} to the
unique run with $n-1$ \emph{ones}. Then $\cS'$ is a half de Bruijn
sequence.
\end{cor}

\subsection{The Construction}

For two sequences $\cT = (t_1,\ldots,t_K)$ and $\cS =
(s_1,\ldots,s_N)$ the product $\cT \otimes \cS$ is a $K \times N$
array $G$ in which  $g_{ij}$, $1 \leq i \leq K$, $1 \leq j \leq
N$, contains the value $t_i \oplus s_j$ (where $\oplus$ denotes
modulo 2 addition, also known as the XOR operator).

Take a half de Bruijn sequence $\cT = (t_1,\ldots,t_K)$ and a de
Bruijn sequence $\cS = (s_1,\ldots,s_N)$ of orders $k$ and $n$,
and lengths $K=2^{k-1}$ and $N=2^n$, respectively, and let $G= \cT
\otimes \cS$. Clearly, each row in $G$ equals either $\cS$ or
$\bar{\cS}$. Similarly, each column equals either $\cT$ or
$\bar{\cT}$. Thus, each row and each column retain their window
property and can serve for self-location in each dimension.
\begin{theorem}
\label{thm:cross} Each cross shaped pattern with $k$ vertical and
$n$ horizontal entries appears exactly once as a pattern in the
array $G$.
\end{theorem}
\begin{proof}
Let $X$ be a column vector of length $k$ and $Y$ be a row vector
of length $n$. Either $X$ or $\bar{X}$ appears in the sequence
$\cT$. Let $\cX$ the pattern which appears. Both $Y$ and $\bar{Y}$
appear in the sequence $\cS$. Crosses with vertical vector $X$ and
horizontal vector $Y$ appear in $G$ only in the portions related
to $Z_1 \deff \cX \otimes Y$ and $Z_2
\deff \cX \otimes \bar{Y}$. Moreover, the crosses in
$Z_1$ are complements of the crosses $Z_2$. For each cross inside
$Z_1$ and $Z_2$ there are two possible assignments, depending on
the mutual entry of the vertical and horizontal component. Each
one of these values appears in either $Z_1$ or $Z_2$.
\end{proof}
By Theorem~\ref{thm:cross}, we can use a cross sensor array to
sample $k$ vertical and $n$ horizontal pixels (with one mutual
pixel) in order to obtain self-location.
\begin{cor}
The proposed method is optimal in terms of the number of sampled
pixels required to achieve self-location with a cross of vertical
length $k$ and horizontal length $n$.
\end{cor}
\begin{cor}
\label{cor:2Dwin} In the array $G$ each sampled sub-array of size
$k \times n$ has a unique location.
\end{cor}

\begin{remark}
Similar and more sophisticated product constructions to generate
arrays with low redundancy and effective two-dimensional
error-correction capabilities, were suggested in various papers,
e.g.~\cite{BBZS,EtYa09}.
\end{remark}
\begin{remark}
In practice, the planar domain is generally not cyclic. In order
to retain the ability to sense all $2^{k-1} \times 2^{n}$ possible
locations with a sensor whose footprint is $k \times n$ pixels
array, we extend $\cT$ and $\cS$ by appending their first $k-1$
and $n-1$ bits, respectively, to their ends. The result is now a
$(2^{k-1}+k-1) \times (2^{n}+n-1)$ array.
\end{remark}

\begin{example}
An example of our proposed two-dimensional grid pattern can be
seen in Fig.~\ref{fig:Grid4_4}. It was generated using a de Bruijn
sequence of order $4$ in the horizontal axis, and a half de Bruijn
sequence of order $5$ in the vertical axis, resulting in a cyclic
array of $16 \times 16$ pixels. The first column and the first row
in the figure contain the location indexes. The second column and
the second row contain $\cT$ and $\cS$, respectively. From the bit
values inside the grid we can decode our position. An example of a
sensor readout is marked in the table. The sensor is a $5$ by $4$
cross. The vertical readout is $10010$, and the horizontal readout
is $1000$ and its unique position can be easily decoded from $\cT$
and $\cS$.
\end{example}
\begin{figure*}[htbp]
%\begin{figure*}[htbp]
\begin{center}
\begin{scriptsize}
\begin{tabular}{|c|c|cccccccccccccccc|}
%\begin{tabular}{|c|c|p{0.5em}p{0.5em}p{0.5em}p{0.5em}p{0.5em}p{0.5em}p{1em}p{0.5em}p{0.5em}p{0.5em}p{0.5em}p{0.5em}p{0.5em}p{0.5em}p{0.5em}p{0.5em}|}

\hline
   & & 1 & 2 & 3 & 4 & 5 & 6 & 7 & 8 & 9 & 10 & 11 & 12 & 13 & 14 & 15 & 16 \\
\hline
   &   & 0 & 0 & 0 & 0 & 1 & 1 & 1 & 1 & 0 & 1 & 1 & 0 & 0 & 1 & 0 & 1 \\
\hline
 1 & 1 & 1 & 1 & 1 & 1 & 0 & 0 & 0 & 0 & 1 & 0 & 0 & 1 & 1 & 0 & 1 & 0 \\
 2 & 1 & 1 & 1 & 1 & 1 & 0 & 0 & 0 & 0 & 1 & 0 & 0 & 1 & 1 & 0 & 1 & 0 \\
 3 & 1 & 1 & 1 & 1 & 1 & 0 & 0 & 0 & 0 & 1 & 0 & 0 & 1 & 1 & 0 & 1 & 0 \\
 4 & 1 & 1 & 1 & 1 & 1 & 0 & 0 & 0 & 0 & 1 & 0 & 0 & 1 & 1 & 0 & 1 & 0 \\
 5 & 1 & 1 & 1 & 1 & 1 & 0 & 0 & 0 & 0 & 1 & 0 & 0 & 1 & 1 & 0 & 1 & 0 \\
 6 & 0 & 0 & 0 & 0 & 0 & 1 & 1 & 1 & 1 & 0 & 1 & 1 & 0 & 0 & 1 & 0 & 1 \\
 7 & 1 & 1 & 1 & 1 & 1 & 0 & 0 & 0 & 0 & 1 & 0 & 0 & 1 & 1 & 0 & 1 & 0 \\
 8 & 0 & 0 & 0 & 0 & 0 & \textcircled{\raisebox{-0.9pt}{1}} & 1 & 1 & 1 & 0 & 1 & 1 & 0 & 0 & 1 & 0 & 1 \\
 9 & 1 & 1 & 1 & 1 & 1 & \textcircled{\raisebox{-0.9pt}{0}} & 0 & 0 & 0 & 1 & 0 & 0 & 1 & 1 & 0 & 1 & 0 \\
10 & 1 & 1 & 1 & 1 & \textcircled{\raisebox{-0.9pt}{1}} & \textcircled{\raisebox{-0.9pt}{0}} & \textcircled{\raisebox{-0.9pt}{0}} & \textcircled{\raisebox{-0.9pt}{0}} & 0 & 1 & 0
& 0 & 1 & 1
   & 0 & 1 & 0 \\
11 & 0 & 0 & 0 & 0 & 0 & \textcircled{\raisebox{-0.9pt}{1}} & 1 & 1 & 1 & 0 & 1 & 1 & 0 & 0 & 1 & 0 & 1 \\
12 & 1 & 1 & 1 & 1 & 1 & \textcircled{\raisebox{-0.9pt}{0}} & 0 & 0 & 0 & 1 & 0 & 0 & 1 & 1 & 0 & 1 & 0 \\
13 & 1 & 1 & 1 & 1 & 1 & 0 & 0 & 0 & 0 & 1 & 0 & 0 & 1 & 1 & 0 & 1 & 0 \\
14 & 1 & 1 & 1 & 1 & 1 & 0 & 0 & 0 & 0 & 1 & 0 & 0 & 1 & 1 & 0 & 1 & 0 \\
15 & 0 & 0 & 0 & 0 & 0 & 1 & 1 & 1 & 1 & 0 & 1 & 1 & 0 & 0 & 1 & 0 & 1 \\
16 & 0 & 0 & 0 & 0 & 0 & 1 & 1 & 1 & 1 & 0 & 1 & 1 & 0 & 0 & 1 & 0 & 1 \\
\hline
\end{tabular}
\end{scriptsize}
\end{center}
\caption{The $16 \times 16$ product array of $\cT \otimes \cS$.
The marked cells illustrate a readout by a cross-shaped sensor.}
\label{fig:Grid4_4}
\end{figure*}

%\subsection{Readout Optimality}
%The proposed method is optimal in terms of the number of sampled
%pixels required to achieve self-location.  Given orders $k$ and
%$n$, the lengths of $\cS$ and $\cT$ are $2^{k-1}$ and $2^n$,
%respectively, yielding a cyclic array of $2^{k+n-1}$ pixels in
%total.  As we said above, we propose a sensor assembly containing
%$k$ vertical and $n$ horizontal pixels with one mutual pixel, or
%$k+n-1$ pixels in total, which is the minimum number of pixels
%required to encode all $2^{k+n-1}$ possible locations on the
%array.

%\subsection{An Acyclic Implementation}
%In practice, the planar domain is generally not cyclic. In order
%to retain the ability to sense all $2^{k-1}$ by $2^{n}$ possible
%locations with a sensor whose footprint is $k$ by $n$ pixels, we
%extend $\cT$ and $\cS$ by appending their first $k-1$ and $n-1$
%bits, respectively, to their ends. The extended array dimensions
%are now $2^{k-1}+k-1$ by $2^{n}+n-1$, respectively.

\subsection{Computing the Location}\label{sec:computing}

The first step in our method recovers the one-dimensional
subsequences that correspond to the location in each dimension.
Essentially, the two-dimensional problem is now reduced to two
independent one-dimensional decoding problems. Decoding the
location of a subsequence in a de Bruijn sequence is a well-known
problem. Decoding of a half de Bruijn sequence is done similarly.

A classic approach of creating a de Bruijn sequence $\cS$ of order
$n$, requires $\mathbf{O}(n)$ space and $\mathbf{O}(n \cdot 2^n)$
time to generate the whole sequence $\cS$~\cite{Fre82,EtLe84}.
This involves $\mathbf{O}(n)$ space and $\mathbf{O}(n \cdot 2^n)$
time, with $n$ being the order of the de Bruijn sequence. If
running time is an issue, one could create and store in advance a
look-up table which lists the locations of all subsequences. This
yields $\mathbf{O}(n)$ time complexity, but requires
$\mathbf{O}(n\cdot 2^n)$ space for the table. For larger $n$, a
more flexible trade-off between time and space complexity was
suggested in~\cite{petriu1988}. A partial look-up table of evenly
spaced locations called \emph{milestones} is created in advance.
During runtime, the algorithm which generates the sequence is
initialized with the query subsequence and then iterated until one
of the milestones is encountered. For example, this can yield
$\mathbf{O}(n \cdot 2^{\frac{n}{2}})$ time complexity and will
require $\mathbf{O}(n \cdot 2^{\frac{n}{2}})$ space for the table.

In either case, implementation of the self location process using
modern computer systems is feasible, at least for reasonable and
practical values of $n$, depending on the application.  Take
$n=16$ for a concrete example. It allows a definition of $2^{16}$
locations, e.g., a resolution of $0.1 \mathrm{mm}$ over a range of
about $6.5$ meters.  In the first approach it would take, in the
worst case, about $65\mathrm{k}$ simple iterations (on a $16$-bit
register), which can be performed reasonably quickly on current
modest embedded processors currently clocked at about tens or
hundreds of Megahertz.  In the second approach, the look-up table
would consume about $128\mathrm{k}$ bytes (each entry being a
two-byte word), which is, again, a quite modest requirement given
today's memory capabilities.

There are more efficient methods to generate de Bruijn
sequences~\cite{MEP96} which can be used in case of an application
in which $k$ and $n$ are much larger. The problem of decoding
perfect maps was considered for example in~\cite{MiPa94}. A comprehensive
survey on this topic was given in~\cite{BuMi91}.

\section{Robust Self-Location}
\label{sec:robust}

%Until now we considered a system with no sensing faults, and used
%the minimum number of $k+n-1$ bits (read in a cross-shaped
%pattern) required for decoding the location.  Recall that not only
%the cross-shaped pattern is unique in the array, but also the
%entire $k \times n$ rectangular sub-array is unique. Thus, if we
The cross shaped sensor is rather `spread out', so it might be a
disadvantage in applications. In this section we show that this
weakness becomes an advantage for robust self-location when the
sensor is of a rectangular shape. If we use a $k \times n$ pixel
sensor (see Corollary~\ref{cor:2Dwin}), we can utilize the
inherent redundancy within the $kn$ bits to decode the location
while overcoming a considerable number of faulty bits readings.
This is also a very practical choice, considering that
two-dimensional rectangular sensor grids are the most common
variety and are the standard choice for most applications.

%The redundancy is clearly seen in the pattern, where all rows
%(resp.~columns) are either identical or inverted copies of one
%another. Mathematically speaking, in the absence of noise, the
%readout matrix has rank $1$.

We assume that less than quarter of the bits in each row and less
than half of the bits in each column of the input array are in
error. As it will be shown in the sequel, this is a fair
assumption which can account for quite strong noise in practical
terms. The algorithm for robust self-location presented in
Fig.~\ref{fig:algorithm} is a simple majority decoding.

\begin{figure}[hbt]
\centering
\begin{algorithm}
\noindent \textit{\bf {Robust self-location algorithm}}

The algorithm's input is a rectangle $Z \deff \{ z_{ij} ~:~ 1\leq
i \leq k, 1 \leq j \leq n \} =\left( X \otimes Y \right) \oplus
\cE$; where $X$ is a vertical $k$-tuple of a vertical half de
Bruijn sequence $\cT$; $Y$ is a horizontal $n$-tuple of a
horizontal de Bruijn subsequence $\cS$; and $\cE$ is a $k \times
n$ error pattern. We assume that less than $\frac{n}{4}$ of the
bits in each row of $\cE$ are \emph{ones} and less than
$\frac{k}{2}$ of the bits in each column of $\cE$ are \emph{ones}.
The output is the original horizontal and vertical subsequences
$X$ and $Y$, respectively.

\begin{itemize}
\item Assume that the first bit of $X$ is $b$. Let $D$ be the
first row of $Z$.

\item For each row $A$ of $Z$
\begin{itemize}
\item if more than half of the bits of $A \oplus D$ are
\emph{zeroes} then the corresponding bit of $X$ is $b$

\item otherwise, the corresponding bit of $X$ is $\bar{b}$.
\end{itemize}
\item Assign 0 or 1 to $b$ to obtain $X$ which appears in $\cS$.

\item For each column $B$ of $Z$
\begin{itemize}
\item if more than half of the bits of $B \oplus X$ are \emph{zeroes} then
the corresponding bit in $Y$ is a \emph{zero}

\item otherwise, the corresponding bit in $Y$ is an \emph{one}.
\end{itemize}
\end{itemize}
\end{algorithm}
\caption{The robust self-location algorithm.}\label{fig:algorithm}
\end{figure}

\begin{theorem}
\label{thm:error_correct} Given a grid of size $2^{k-1} \times
2^n$ and a $k \times n$ pixel sensor, if less than quarter of the
bits in each row and less than half of the bits in each column are
in error, then the algorithm accurately decodes the sensor
location.
\end{theorem}
\begin{proof}
Since the number of errors in a row is less than~$\frac{n}{4}$ it
follows that two rows which were originally the same will agree in
more than half of their bits and two complement rows will disagree
in more than half of their bits. Therefore, the related bits of
$X$ will be the same or different, respectively. Having all the
$k$ bits of $X$ in terms of the variable $b$, there is only one
assignment of a legal $k$-tuple since the vertical sequence is a
half de Bruijn sequence.

Having the correct vertical subsequence $X$ and since the number
of errors in a column is less than $\frac{n}{2}$ it follows that
if $X$ agree in more than $\frac{n}{2}$ bits with a column then
the corresponding bit of $Y$ is a \emph{zero}; and if it disagree
in more than $\frac{n}{2}$ bits with a column then the
corresponding bit of $Y$ is a \emph{one}.
\end{proof}

\begin{remark}
\label{rem:distance} Decoding can be done also if more than
quarter of the bits in some rows are in error. A slightly better
condition would be to require that the number of distinct
positions in error in any two rows is less than $\frac{n}{2}$.
This requirement can be further improved.
\end{remark}

Similar algorithm will also work if we will exchange between rows
and columns, or equivalently if we will consider a transposed
array. Therefore, we can exchange our assumption on the number of
wrong bits in a row or a column. But, having for example at least
half of the bits wrong in a given column (or a given row) will
cause a wrong identification of the original subsequences.

\begin{lemma}
\label{lem:halfR} Given a grid of size $2^{k-1} \times 2^n$ and a
$k \times n$ pixel sensor, if at least half of the bits in one of
the columns of a pixel sensor are in error, then we cannot ensure
accurate decoding of the original subsequences.
\end{lemma}
\begin{proof}
Let $X$ and $Y$ two $n$-tuples which differ only in the first bit.
Both $X$ and $Y$ appears as a window of length $n$ in the de
Bruijn sequence $\cS$ of order $n$. Let $Z$ be a $k$-tuple which
appears as a window in the sequence $\cT$. The products $Z \otimes
X$ and $Z \otimes Y$ appears as $k \times n$ windows in the array
$\cT \otimes \cS$. Both $k \times n$ windows differ only in the
first column and it would be impossible to distinguish between the
two windows if half of the bits in the first column are in error.
If more than half of the bits in the first column are in error
then a wrong decoding of the sensor location will be made. The
same arguments can be applied to any other column.
\end{proof}
We note that by Lemma~\ref{lem:halfR} we cannot correct
$\left\lceil \frac{k}{2} \right\rceil$ or more random errors in a
$k \times n$ array. The reason is that the array is highly
redundant. This is quite weak from an error-correction point of
view. But, by Theorem~\ref{thm:error_correct} we are able to
correct about $\frac{kn}{4}$ errors in an $k \times n$ array if
less than $\frac{n}{4}$ errors occur in a row and less than
$\frac{k}{2}$ errors occur in a column. The reason is that
redundant rows and columns are used for the majority decoding.
This result is quite strong from error-correction point of view.
Thus, the weakness for one type of errors becomes an advantage for
another type of errors.

\begin{example}
The $7 \times 9$ sub-array of Fig.~\ref{grid2} has no more than
two errors in a row and three errors in a column. The first row
has more than half bits in common with the 5th and the 7th rows.
Thus the vertical pattern is $b\bar{b}\bar{b}\bar{b}b\bar{b}b$.
Suppose that $b=1$, i.e. the vertical pattern is $1000101$. We now
compare all of the columns with $1000101$.  If more than half of
the corresponding bits agree, the bit in the horizontal sequence
is \emph{one}; otherwise it is a \emph{zero}. Thus, the sequence
is $110111001$.  The sub-array with no errors is presented in
Fig.~\ref{grid1}.
\end{example}

\renewcommand{\arraystretch}{0.8}
\begin{figure}
\begin{scriptsize}
\centerline{
$\begin{array}{ccccccccc}
1 & 0 & 0 & 1 & 0 & 1 & 0 & 0 & 1  \\
0 & 0 & 0 & 0 & 0 & 1 & 1 & 1 & 0  \\
1 & 0 & 1 & 0 & 0 & 0 & 1 & 1 & 1  \\
0 & 0 & 1 & 0 & 0 & 0 & 1 & 0 & 0  \\
1 & 1 & 0 & 0 & 1 & 0 & 0 & 0 & 1  \\
0 & 0 & 1 & 0 & 1 & 0 & 1 & 1 & 0  \\
1 & 0 & 0 & 1 & 1 & 0 & 0 & 0 & 1
\end{array} $
}
\end{scriptsize}
\caption{A $7 \times 9$ sub-array with errors.} \label{grid2}

\vspace{0.5cm}

\renewcommand{\arraystretch}{0.8}
%\begin{figure}[ht]
\begin{scriptsize}
\centerline{ $\begin{array}{ccccccccc}
1 & 1 & 0 & 1 & 1 & 1 & 0 & 0 & 1  \\
0 & 0 & 1 & 0 & 0 & 0 & 1 & 1 & 0  \\
0 & 0 & 1 & 0 & 0 & 0 & 1 & 1 & 0  \\
0 & 0 & 1 & 0 & 0 & 0 & 1 & 1 & 0  \\
1 & 1 & 0 & 1 & 1 & 1 & 0 & 0 & 1  \\
0 & 0 & 1 & 0 & 0 & 0 & 1 & 1 & 0  \\
1 & 1 & 0 & 1 & 1 & 1 & 0 & 0 & 1
\end{array} $
}
\end{scriptsize}
\caption{The corrected sub-array.} \label{grid1}
\end{figure}
\renewcommand{\arraystretch}{1}

%\renewcommand{\arraystretch}{0.8}
%\begin{figure}
%\begin{tiny}
%\centerline{ $\begin{array}{ccccccccccccccccccc} 1 & 0 & 0 & 1 & 0
%& 1 & 0 & 0 & 1 & {\bold |} & 1 & 1 & 0 & 1 & 1 & 1 & 0 & 0 & 1  \\
%0 & 0 & 0 & 0 & 0 & 1 & 1 & 1 & 0 & {\bf |} & 0 & 0 & 1 & 0 & 0 &
%0 & 1
%& 1 & 0  \\
%1 & 0 & 1 & 0 & 0 & 0 & 1 & 1 & 1 & {\bf |} & 0 & 0 & 1 & 0 & 0 & 0 & 1 & 1 & 0  \\
%0 & 0 & 1 & 0 & 0 & 0 & 1 & 0 & 0 & {\bf |} & 0 & 0 & 1 & 0 & 0 & 0 & 1 & 1 & 0  \\
%1 & 1 & 0 & 0 & 1 & 0 & 0 & 0 & 1 & {\bf |} & 1 & 1 & 0 & 1 & 1 & 1 & 0 & 0 & 1  \\
%0 & 0 & 1 & 0 & 1 & 0 & 1 & 1 & 0 & {\bf |} & 0 & 0 & 1 & 0 & 0 & 0 & 1 & 1 & 0  \\
%1 & 0 & 0 & 1 & 1 & 0 & 0 & 0 & 1 & {\bf |} & 1 & 1 & 0 & 1 & 1 &
%1 & 0 & 0 & 1
%\end{array} $
%}
%\end{tiny}
%\caption{A $7 \times 9$ corrupted sub-array (left) and the
%corrected sub-array (right).} \label{grid}
%\end{figure}
%\renewcommand{\arraystretch}{1}

%\subsection{Probabilistic Analysis}

Now, we analyze the error rates in individual bits that
allow us to determine the probability that the position is
determined correctly. Given a $k \times n$ rectangle in which the
probability of each bit being correct is $p$ independent of the
other bits, we can determine the probability that each row
satisfies the condition above, that less than quarter of the bits are
in error.  Then the probability that each row is satisfactory is
the individual row probability raised to the $k^{\mbox{th}}$
power, the number of rows. For simplicity we assume now that $k=n$.

To find the probability that all of the rows satisfy the row
condition, we raise the probabilities $P(n; p)$ to the power of
the number of rows. These are given in Table~\ref{probtable2}.

\begin{table}[h]
\begin{center}
\begin{tabular}{c|cccc}
$p \setminus n$ & 8 & 16 & 32 & 64 \\ \hline
0.90  & 0.191 & 0.322 & 0.687 & 0.9858 \\
0.91  & 0.253 & 0.443 & 0.817 & 0.9957 \\
0.92  & 0.329 & 0.573 & 0.906 & 0.9989 \\
0.93  & 0.417 & 0.699 & 0.959 & 0.9998 \\
0.94  & 0.517 & 0.809 & 0.985 & $> 0.9999$ \\
0.95  & 0.624 & 0.894 & 0.996 & $> 0.9999$ \\
0.96  & 0.733 & 0.951 & 0.9991 & $> 0.9999$ \\
0.97  & 0.835 & 0.982 & 0.9999 & $> 0.9999$ \\
0.98  & 0.920 & 0.9962 & $> 0.9999$ & $> 0.9999$ \\
0.99  & 0.979 & 0.9997 & $> 0.9999$ & $> 0.9999$ \vspace*{12pt} \\
\end{tabular}
\caption{Probabilities that each of the rows of an $n \times n$
window has fewer than quarter of its bits in error when the
probability that each bit is correct is $p$.} \label{probtable2}
\end{center}
\end{table}

In order to have a probability of at least 0.99 that the row
condition is satisfied, we need $p > 0.994$ for $n = 8$, $p >
0.98$ for $n = 16$, $p > 0.95$ for $n = 32$, and $p > 0.91$ for $n
= 64$.  In order for the row condition to be satisfied with
probability at least $0.999$, it is sufficient that $p > 0.99$ for
$n = 16$, $p > 0.96$ for $n = 32$, and $p > 0.93$ for $n = 64$.
Also, if $p > 0.98$ for $n = 32$ or $p > 0.94$ for $n = 64$, the
row condition is satisfied with probability greater than 0.9999 .

The probability that the column condition (that less than half the
bits are in error) is not satisfied when the row condition is
satisfied is negligible.  For example, if we let $Q(n; p)$
represent the probability that the column condition is not
satisfied, i.e. \begin{equation} \label{prob3} Q(n; p) =
\sum_{i = \lceil \frac{n}{2} \rceil}^n \left(\begin{array}{c} n \\
i
\end{array} \right) p^i (1 - p)^{n-i} \;\;, \end{equation} then we have results
such as $Q(16; 0.99) = 1.2 \times 10^{-12}$. In a square array,
the probability that the column condition is not satisfied for at
least one of the columns is then $1 - (1 - 1.2 \times
10^{-12})^{16} = 1.9 \times 10^{-11}$.

%Fig.~\ref{fig:match_graph} shows how the probability of the
%location being determined correctly varies with the probability of
%each individual bit being correct for various window sizes.
%
%%\begin{figure}[htbp]
%%\begin{center}
%%\includegraphics[width=1\textwidth]{success_graph0.ps}
%%\caption{success}
%%\end{center}
%%\end{figure}
%
%
%\begin{figure}[htb]
%\begin{center}
%\includegraphics[width=0.5\textwidth]{match_graph0.pdf}
%\end{center}
%\caption{Probability of correct match for varying window sizes and
%probabilities for each individual bit being correct.}
%\label{fig:match_graph}
%\end{figure}

\section{Conclusion and Future Research}
\label{sec:conclude}

Implementing absolute self-location in a planar region using
special patterns is a viable and proven approach and can solve a
variety of technological problems. In this paper we proposed a
solution based on robust two-dimensional arrays with a
two-dimensional window property. The method also has a rather
strong error-correction capability. It enables to correct errors
if less than quarter of the bits in a row and less than half of
the bits in a column are in error.

In some applications, the alignment of the sensor array to the
grid pattern is not guaranteed. The sensor may be arbitrarily
translated and rotated, so that retrieving the local bit matrix is
not trivial. Position location of one-dimensional sequences, when
the orientation of the subsequence is not known was considered
in~\cite{DMRW}. The solution in two-dimensional arrays is to
sample the region at a somewhat higher resolution than $k$ by $n$,
and analyze the image in order to first estimate the pose of the
pattern of rows and columns. Since the proposed pattern has a very
pronounced structure consisting of identical or inverted rows (as
well as columns), this can greatly aid in the task. Using an
M-sequence and its complement as vertical and horizontal sequences
in our construction can also help in solving of the orientation
problem. A complete analysis of these issues is a problem for
future research.

There are many other future research problems in this area. Some
related to our specific construction and some are to new possible
construction methods.
\begin{enumerate}
\item As indicated in Remark~\ref{rem:distance}, the claim in
Theorem~\ref{thm:error_correct} can be strengthened. What is the
strongest claim on the error capability of our scheme? Do the de
Bruijn sequence and the half de Bruijn sequence that we selected
have any influence on this claim?

\item How can we improve the error-correction capabilities of our
scheme if the de Bruijn sequence and the half de Bruijn sequence
are derived from M-sequences with error-correction capabilities as
indicated in~\cite{KuWe92}.

\item The array obtained by our method can correct a limited
number of random errors, even so we proved that the probability
for such errors which the method cannot correct is negligible.
Generating arrays with window properties which can correct large
number of random errors is an important topic for future research.

\item Finally, we note that a folding method for generating
pseudo-random arrays from M-sequences was suggested
in~\cite{McSl76}. This method was subsequently generalized
in~\cite{Etz11}. Can this method be adapted also to generate
better pseudo-random arrays which can correct random errors? Using
the M-sequences as suggested by~\cite{KuWe92} for this purpose
could be the first step in attempting to find an answer to such
questions..
\end{enumerate}

%%*******************************************************************************
%%*                                                                             *
%%*                          Acknowledgment                                     *
%%*                                                                             *
%%*******************************************************************************
\section*{Acknowledgment}

The authors would like to thank the anonymous reviewers for their
very useful comments which greatly helped us to improve our paper.

%
%*************************** References *********************************
%
%\bibliography{allbib,extra}

\end{document}